\pdfoutput=1 
\documentclass{article}[11pt]
\usepackage{geometry}
\usepackage{amssymb}
\usepackage{amsthm}
\usepackage{amsmath}
\usepackage{bbm}
\usepackage{xcolor}
\usepackage{hyperref}
\usepackage[hyphenbreaks]{breakurl}
\usepackage{microtype}
\usepackage{breakcites}
\definecolor{ceruleanblue}{rgb}{0.16, 0.32, 0.75}
\definecolor{darkmidnightblue}{rgb}{0.0, 0.2, 0.4}
\definecolor{darkpastelgreen}{rgb}{0.01, 0.75, 0.24}
\definecolor{bleudefrance}{rgb}{0.19, 0.55, 0.91}
\hypersetup{
	colorlinks   = true,
	citecolor    = bleudefrance,
	linkcolor	  = darkpastelgreen,
	urlcolor     = bleudefrance
}
\usepackage[capitalize, nameinlink]{cleveref}

\RequirePackage[T1]{fontenc} \RequirePackage[tt=false, type1=true]{libertine} \RequirePackage[varqu]{zi4} \RequirePackage[libertine]{newtxmath}

\usepackage{verbatim}

\newcommand{\bv}[1]{\mathbf{#1}}
\newcommand{\Var}{\operatorname{Var}}
\newcommand{\diag}{\operatorname{diag}}

\DeclareMathOperator*{\E}{\mathbb{E}}
\DeclareMathOperator*{\R}{\mathbb{R}}
\DeclareMathOperator*{\argmin}{arg\,min}

\DeclareMathOperator{\tr}{tr}

\makeatletter
\newtheorem*{rep@theorem}{\rep@title}
\newcommand{\newreptheorem}[2]{%
	\newenvironment{rep#1}[1]{%
		\def\rep@title{\Cref{##1}}%
		\begin{rep@theorem}}%
		{\end{rep@theorem}}}
\makeatother

\newtheorem{theorem}{Theorem}
\newreptheorem{theorem}{Theorem}
\newtheorem{lemma}[theorem]{Lemma}
\newreptheorem{lemma}{Lemma}
\newtheorem{definition}{Definition}

\newtheorem{corollary}[theorem]{Corollary}

\newtheorem{fact}[theorem]{Fact}

  \usepackage{nth}
  \usepackage{intcalc}

  \newcommand{\cAAAI}[1]{AAAI\ Conference\ on\ Artificial (AAAI)}

\title{A Tight Analysis of Hutchinson's Diagonal Estimator}
\author{
	Prathamesh Dharangutte \\ Rutgers University\\ \texttt{prathamesh.d@rutgers.edu}
	\and 
	Christopher Musco\\ New York University\\ \texttt{cmusco@nyu.edu}
}
\date{\today}

\begin{document}
\maketitle
\begin{abstract}
	Let $\bv{A}\in \R^{n\times n}$ be a matrix with diagonal $\diag(\bv{A}) \in \R^n$. We show that the simple and practically popular Hutchinson's estimator, run for $m$ trials, returns a diagonal estimate $\tilde{d}\in \R^n$ such that with probability $1-\delta$, 
	\begin{align*}
		\|\tilde{d} - \diag(\bv{A})\|_2 \leq c\sqrt{\frac{\log(2/\delta)}{m}}\|\bar{\bv{A}}\|_F.
	\end{align*}
	Above $c$ is a fixed constant 
	and $\bar{\bv{A}}$ equals $\bv{A}$ with its diagonal set to zero.
	This result improves on recent work in \cite{BastonNakatsukasa:2022} by a $\log(n)$ factor, yielding a bound that is independent of the matrix dimension, $n$.
	We show a similar bound for variants of Hutchinson's estimator that use non-Rademacher random vectors. 
	
\end{abstract}

\section{Introduction}
We give a short and tight analysis of the popular Hutchinson's estimator for approximating the diagonal of a square matrix, $\bv{A}$, given only \emph{implicit} matrix-vector multiplication access to the matrix \cite{Hutchinson:1990,BekasKokiopoulouSaad:2007}. 
\begin{definition}[Hutchinson's Diagonal Estimator] 
	\label{def:diag_est}
	Let $\bv{g}^{1},\ldots, \bv{g}^{m} \in \{-1,+1\}^n$ be independent random vectors, each with i.i.d. Rademacher (random $\pm 1$) entries. Hutchinson's diagonal estimator $\bv{r}^m(\bv{A})  \in \R^n$ is:
	\begin{align*}
		\bv{r}^m(\bv{A}) = \frac{1}{m}\sum_{z=1}^m \bv{g}^{z}\odot \bv{A}\bv{g}^{z},
	\end{align*}
where $\bv{a}\odot \bv{b} \in \R^n$ denotes the Hadamard product (entrywise product) between vectors $\bv{a}\in \R^n$ and $\bv{b}\in \R^n$. 
\end{definition}
Computing $\bv{r}^m(\bv{A})$ requires $m$ matrix-vector multiplications with $\bv{A}$, and it is not hard to check that it is an unbiased estimator for the diagonal, i.e., $\E[\bv{r}^m(\bv{A})] = \diag(\bv{A})$, where $\diag(\bv{A})$ is a vector containing $\bv{A}$'s diagonal elements. 
Hutchinson's diagonal estimator is simple to implement and is widely applied across applications in computational science \cite{AsterBorchersThurber:2019,MetivierBretaudeauBrossier:2014}, machine learning \cite{Pavlo-Molchanov:2017,ErikssonDongLee:2018}, and optimization \cite{YaoGholamiShen:2021,DauphinVriesBengio:2015}. In these applications, it is used to estimate the diagonals of large Hessian matrices, matrix inverses, and other matrices that are expensive to construct explicitly, but for which matrix-vector multiplications can be implemented quickly (e.g. using backpropagation or an iterative linear system solver).

However, despite its popularity, there has been a lack of theoretical work on Hutchinson's diagonal estimator, and in particular on the question of how large $m$ should be so that $\bv{r}^m(\bv{A})$ concentrates around its expectation. This is in contrast to the closely related Hutchinson's  \emph{trace} estimator, which has been heavily studied and for which a tight analysis is known \cite{Roosta-KhorasaniAscher:2015,CortinovisKressner:2022,MeyerMuscoMusco:2021,WoodruffZhangZhang:2022}.

Two recent papers do provide bounds for the diagonal estimation problem \cite{HallmanIpsenSaibaba:2022,BastonNakatsukasa:2022}. The second proves that if $m = O(\log(n/\delta)/\epsilon^2)$, then with probability $1-\delta$, 	$\|\bv{r}^m(\bv{A}) - \diag(\bv{A})\|_2 \leq \epsilon \|\bar{\bv{A}}\|_F$, where $\|\bar{\bv{A}}\|_F^2 = \|\bv{A}\|_F^2 - \|\diag(\bv{A}) \|_2^2$ denotes the squared Frobenius norm of $\bv{A}$ with its diagonal entries set to $0$. 
Our goal is to tighten the analysis of \cite{BastonNakatsukasa:2022} by removing the $\log(n)$,  i.e., to prove that to achieve error $\epsilon \|\bar{\bv{A}}\|_F$, just  $m = O\left(\log(1/\delta)/\epsilon^2\right)$ matrix-vector products are necessary. Formally, we prove:
\begin{theorem}[Main Theorm]
	\label{thm:main}
 For any $\delta \in (0,1]$ and $m\geq 1$, with probability $1-\delta$:
	\begin{align*}
		\|\bv{r}^m(\bv{A}) - \diag(\bv{A})\|_2 \leq c\sqrt{\frac{\log(2/\delta)}{m}}\|\bar{\bv{A}}\|_F,
	\end{align*}
where $c$ is an absolute constant independent of $\bv{A}$ and all other problem parameters. 
\end{theorem}
The dependence on $\log(n)$ in the analysis of \cite{BastonNakatsukasa:2022} arises through the use of a union bound argument: they show that  Hutchinson's estimator separately obtains an accurate estimate for each entry of $\bv{A}$'s diagonal, and thus $\|\bv{r}^m(\bv{A}) - \diag(\bv{A})\|_2$ can be bounded\footnote{It is possible to replace $n$ with a natural ``intrinsic dimension'' parameter that is smaller for some problem instances  \cite{Hallman:2022}.}. A similar $\log(n)$ appeared in early analysis for the trace estimation problem \cite{AvronToledo:2011} and was later removed \cite{Roosta-KhorasaniAscher:2015}. 
We obtain a comparable improvement through a refined analysis of the stochastic diagonal estimator that relies on a symmetrization argument and techniques for proving vector-valued Bernstein inequalities \cite{Yurinskii:1970}. 

We do note that it is possible to obtain a \emph{low probability} result for Hutchinson's estimator which almost matches the bound of \cref{thm:main}, but with a costly linear dependence on $1/\delta$. We discuss this result in Section \ref{sec:relation}. Additionally, using this low probability result, the same asymptotic complexity as \cref{thm:main} (with no $n$ dependence, and just a $\log(1/\delta)$ dependence) can be obtained by combining Hutchinson's estimator with a multi-dimensional variant of the ``median trick''. We discuss this approach in Section \ref{sec:general_high_prob}. However, we are mostly interested in analyzing Hutchinson's estimator itself as the method is 1) simpler to implement 2) essentially parameter free (only requires specifying $m$) and 3) the most widely used diagonal estimator in practice.

We also note that \cref{thm:main} is tight, and the bound cannot be further improved for Hutchinson's estimator. To see that this is the case, consider the matrix $\bv{A} = \begin{bmatrix}0&1\\0&0\end{bmatrix}$. We can check that $r^m(\bv{A}) = \begin{bmatrix}S/m\\0\end{bmatrix}$ where $S$ is a sum of $m$ independent $\pm1$ random variables. By the well-known tightness of the Chernoff bound (see e.g. \cite{KleinYoung:2015}) we will only have that $S/m \leq \epsilon$ with probability $1-\delta$ if $m = O(\log(1/\delta)/\epsilon^2)$, which matches the upper bound implied by  \cref{thm:main}. It is possible that a different estimator could improve on \cref{thm:main}, either in general or for some classes of matrices. Proving a strong lower bound showing the result is optimal in e.g. the matrix-vector product model of computation is a nice open question \cite{SunWoodruffYang:2019,MeyerMuscoMusco:2021}.

\section{Preliminaries}
\noindent\textbf{Notation.} For a vector $\bv{y}\in \R^n$, $\|\bv{y}\|_2 = (\sum_{i=1}^n y_i^2)^{1/2}$ denotes the Euclidean norm. For a matrix $\bv{A}\in \R^{n\times m}$, $\|\bv{A}\|_F = (\sum_{i=1}^n\sum_{j=1}^m A_{ij}^2 )^{1/2}$ denotes the Frobenius norm and $\|\bv{A}\|_2  = \max_{\bv{x}\in \R^m} \|\bv{A}\bv{x}\|_2/\|\bv{x}\|_2$ denotes the spectral norm. When $\bv{A}$ is square, $\tr(\bv{A}) = \sum_{i=1}^n A_{ii}$ denotes the trace. We use $c,c', C,$ etc. to denote absolute constants that are independent of the problem input and all other parameters. The exact value of these constants changes depending on context. 
\vspace{1em}

\noindent{\textbf{Random Variables}.}
When analyzing random variables, we will make use of the properties of sub-Gaussian and sub-exponential random variables, using the notation of \cite{Vershynin:2018}. Formally we define:
\begin{definition}[Sub-Gaussian Random Variable]
	\label{def:subgauss}
	A random variable $X$ is sub-Gaussian with parameter $K$ if we have $\E\left[e^{X^2/K^2 }\right] \leq 2$.
\end{definition}

\begin{definition}[Sub-exponential Random Variable]
		\label{def:subexp}
	A random variable $X$ is sub-exponential with parameter $K$ if we have $\E\left[e^{|X|/K}\right] \leq 2$.
\end{definition}\textit{}

\noindent{\textbf{Trace Estimation}.} To prove \cref{thm:main} we will relate Hutchinson's diagonal estimator (\cref{def:diag_est}) to the well-known Hutchinson's trace estimator, which we define below:
\begin{definition}[Hutchinson's Trace Estimator] 
	\label{def:trace_est}
	Let $\bv{g}\in \{-1,+1\}^n$ be a vectors with i.i.d. Rademacher entries. Hutchinson's trace estimator $T$ for a matrix $\bv{B}\in \R^{n\times n}$ is:
	\begin{align*}
		T(\bv{B}) =  \bv{g}^T \bv{B}\bv{g}.
	\end{align*}
\end{definition}
It is not hard to show that $\E[T(\bv{B})] = \tr(\bv{B})$ and $\Var[T(\bv{B})] = 2\|\bv{\bar{B}}\|_F^2$,  where $\bv{\bar{B}}$ denotes $\bv{B}$ with its diagonal entries set to $0$. By averaging repeated copies of the estimator we can obtain a lower variance estimate. To prove high probability error bounds, a tight analysis can be obtained via the Hanson-Wright inequality, which implies that  $T(\bv{B})$ exhibits exponential concentration \cite{CortinovisKressner:2022}. We will use an intermediate result stated in Section 6.2 of \cite{Vershynin:2018} as a step towards proving Hanson-Wright\footnote{Note that when $\bv{g}$ contains Rademacher random variables, $Z(\bv{B}) = \tr(\bv{B})-T(\bv{B})$ exactly equals $\sum_{i\neq j} g_ig_jB_{ij}$, which is precisely the ``off-diagonal sum'' random variable bounded in \cite{Vershynin:2018}.}:
\begin{lemma}[\cite{Vershynin:2018}]
	\label{lem:hutch_moment_bound}
	Let $Z(\bv{B}) = T(\bv{B})-\tr(\bv{B})$ be the error of Hutchinson's trace estimator as in \cref{def:trace_est}. For absolute constants $c,C$, we have:
\begin{align*}
	\E\left[e^{\lambda Z(\bv{B})}\right] &\leq e^{C\lambda^2 \|\bv{B}\|_F^2} & &\text{for all} & |\lambda| \leq c/\|\bv{B}\|_2.
\end{align*}
\end{lemma}

\subsection{Relation Between Diagonal Estimator and  Trace Estimator}
\label{sec:relation}
Consider $\bv{r}^m(\bv{A})$ and as before let $\diag(\bv{A})$ denote the diagonal of $\bv{A}$. Let $\bv{g}^{1}, \ldots, \bv{g}^{m}$ be the $m$ random $\pm 1$ vectors used to obtain $\bv{r}^m(\bv{A})$. We can rewrite the mean zero random vector $\bv{r}^m(\bv{A}) - \diag(\bv{A})$ as:
\begin{align*}
	\bv{r}^m(\bv{A}) - \diag(\bv{A}) &= \frac{1}{m}\sum_{z=1}^m \bv{e}_z & &\text{where for $z=1, \ldots, m$} & \bv{e}_z = \bv{g}^{z}\odot \bv{A}\bv{g}^{z} - \diag(\bv{A}). 
\end{align*}
Note that the $i^\text{th}$ entry in $\bv{e}_z$ equals $\sum_{j\neq i}A_{ij}g_i^{z}g_j^{z}$. 
Using that $(g_i^z)^2 = 1$ for all $i,z$ and recalling that $\bar{\bv{A}}$ denotes $\bv{A}$ with its diagonal set to zero, we have:
\begin{align*}
	\|\bv{e}_z\|_2^2 = \sum_{i=1}^d \left(\sum_{j\neq i}A_{ij}g_i^{z}g_j^{z}\right)^2 =\sum_{i=1}^d\sum_{j\neq i}\sum_{k \neq i} A_{ij}A_{ik} g_i^z g_i^zg_j^zg_k^z =\sum_{i=1}^d \sum_{j=1}^d \sum_{k=1}^d \bar{{A}}_{ij}\bar{{A}}_{ik} g_j^zg_k^z 
		&= \sum_{j=1}^d \sum_{k=1}^d g_j^zg_k^z \sum_{i=1}^d\bar{{A}}_{ij}\bar{{A}}_{ik}.
\end{align*}
Let $\bv{B} = \bar{\bv{A}}^T\bar{\bv{A}}$. We have that $B_{jk} = \sum_{i=1}^d\bar{{A}}_{ij}\bar{{A}}_{ik}$, so we can rewrite the above as:
\begin{align}
	\label{eq:diag_to_trace}
	\|\bv{e}_z\|_2^2 = \sum_{j=1}^d \sum_{k=1}^d g_j^zg_k^z {B}_{jk} =  {\bv{g}^z}^T \bv{B} \bv{g}^z.
\end{align}
In other words, $\|\bv{e}_z\|_2^2$ is identically distributed to Hutchinson's \emph{trace estimator} applied to the positive semi-definite matrix $\bv{B}$. An immediate consequence of \cref{eq:diag_to_trace} is that $\E\|\bv{e}_z\|_2^2 = \tr(\bv{B}) = \|\bar{\bv{A}}\|_F^2$. This in turn yields the following:
\begin{lemma}[Expected Squared Error of Hutchinson's Diagonal Estimator]
	\label{lem:expectation_lemma}
	Let $\bv{r}^m(\bv{A})$  as in \cref{def:diag_est}. We have:
	\begin{align*}
	\E\left[ \|\bv{r}^m(\bv{A}) - \diag(\bv{A}) \|_2^2\right] = \frac{1}{m}\|\bar{\bv{A}}\|_F^2.
	\end{align*}
\end{lemma}
\begin{proof}
	\begin{align*}
	\E \left[\|\bv{r}^m(\bv{A}) - \diag(\bv{A}) \|_2^2\right] = \E \left[\left\|\frac{1}{m}\sum_{z=1}^m \bv{e}_z \right\|_2^2\right] =  \frac{1}{m^2}\E\left[\sum_{z=1}^m \|\bv{e}_z\|_2^2 +\sum_{z=1}^m\sum_{w\neq z} {\bv{e}_z}^T\bv{e}_w \right] = \frac{1}{m^2}\left[\sum_{z=1}^m\|\bar{\bv{A}}\|_F^2 +0 \right] 
	\end{align*}
In the last inequality we used that ${\bv{e}_z}^T\bv{e}_w = 0$ because the random vectors are mean zero and independent.
\end{proof}
Applying Markov's inequality, an immediate consequence of \cref{lem:expectation_lemma} is that $\|\bv{r}^m(\bv{A}) - \diag(\bv{A}) \|_2^2 \leq \frac{1}{m\delta}\|\bar{\bv{A}}\|_F^2$ with probability $1-\delta$. 
Setting $m = \frac{1}{\epsilon^2\delta}$ we thus have that with probability  $1-\delta$, $\|\bv{r}^m(\bv{A}) - \diag(\bv{A}) \|_2\leq \epsilon\|\bar{\bv{A}}\|_F$.
Notably this simple bound already avoids the $\log(n)$ dependence from \cite{BastonNakatsukasa:2022}, but it incurs a suboptimal $1/\delta$ dependence in comparison to that result and \cref{thm:main}, which depend on $\log(1/\delta)$. 

\section{Proof of Main Theorem}
\label{sec:main_proof}
In this section we prove \cref{thm:main}, which requires bounding the norm of $\bv{r}^m(\bv{A}) - \diag(\bv{A})$. This random vector can be written as the average of $m$ mean-zero random vectors $\bv{r}^m(\bv{A}) - \diag(\bv{A}) = \frac{1}{m}\sum_{z=1}^m \bv{e}_z$. Via the connection to Hutchinson's trace estimator, we know the expected norm of each $\bv{e}_z$ is equal to $\|\bar{\bv{A}}\|_F$. Moreover, each norm should not be much larger than $\|\bar{\bv{A}}\|_F$ with high probability due to the concentration of Hutchinson's trace estimator. 
Thus a natural approach might be to apply a ``vector valued Bernstein'' inequality for sums of norm-bounded random vectors \cite{Yurinskii:1970}. However a direct application of prior work yields a suboptimal polynomial dependence on $\log(1/\delta)$. 

Alternatively, since $\|\bv{r}^m(\bv{A}) - \diag(\bv{A})\|_2^2$ is a low-degree (degree $4$) polynomial in Rademacher random variables, we might hope to prove concentration by applying techniques based on hypercontractivity to bound the random variable's higher moments, as done e.g. in \cite{Skorski:2021} for Hutchinson's trace estimator. However, doing so would require establishing a bound on the second moment $\E[\|\bv{r}^m(\bv{A}) - \diag(\bv{A})\|_2^4]$, which is already challenging. Relying directly on hypercontractivity also seems to be limited to yielding a suboptimal dependence on $\log(1/\delta)$.

We take another approach, providing an analysis that loosely follows the same approach used for Lemma 2 in Yurinskii's proof of the vector-valued Bernstein inequality and does not require $\bv{e}_1, \ldots, \bv{e}_m$ to be strictly bounded.  


\subsection{Symmetrization and Scalar Comparison}
Let $\bv{e} = \sum_{z=1}^m \bv{e}_z$ and note that  $\bv{e} = m\cdot (\bv{r}^m(\bv{A}) - \diag(\bv{A}) )$. Our goal will be to upper bound the moments of $\bv{e}$'s squared norm by the moments of an easier to analyze scalar random variable. 
To do so, we start with a symmetrization argument. First, consider the alternative random vector $\tilde{\bv{e}} = \bv{e}_1 - \bv{e}_2$ where $\bv{e}_1$ and $\bv{e}_2$ are i.i.d. copies of $\bv{e}$. 
Using that $f(\bv{x}) = \|\bv{x}\|_2^{2k}$ is a convex function and that $\E[\bv{e}_1] = \E[\bv{e}_2] = 0$, we can apply Jensen's inequality to show that:
\begin{align}
	\label{eq:moment_bound_1}
	\mathbb{E}\left[\|\bv{e}\|_2^{2k}\right] \leq \mathbb{E}\left[\|\bv{e}_1 - \bv{e}_2\|_2^{2k}\right] = \mathbb{E}\left[\|\tilde{\bv{e}}\|_2^{2k}\right]. 
\end{align}
See Lemma 6.1.2 in \cite{Vershynin:2018} for a detailed derivation of the above inequality. Next, we can turn our attention to bounding $\mathbb{E}\left[\|\tilde{\bv{e}}\|_2^{2k}\right]$. Letting $\bv{e}_{z,1}$ and $\bv{e}_{z,2}$ be i.i.d. copies of $\bv{e}_z$ (i.e. error vectors of Hutchinson's diagonal estimator applied with a single random vector), we have that $\tilde{\bv{e}} = \sum_{z=1}^m \bv{e}_{z,1} - \bv{e}_{z,2}$.
Let $\bv{w}_z$ denote $\bv{w}_z = \bv{e}_{z,1} - \bv{e}_{z,2}$ and note that $\tilde{\bv{e}} = \sum_{z=1}^m \bv{w}_z$. Let $W_z$ be a scalar random variable equal to $r_z\cdot\|\bv{w}_z\|_2$, where $r_z$ is a $\pm 1$ Rademacher random variable. For all $k = 1, 2, \ldots$ we have that:
\begin{align*}
	\mathbb{E}[W_z] &= 0 & &\text{and} & \mathbb{E}[W_z^{2k}] = \mathbb{E}[\|\bv{w}_z\|_2^{2k}].
\end{align*}
Let $\tilde{E} = \sum_{z=1}^m W_z$. We will bound the moments of $\|\tilde{\bv{e}}\|_2^2$ by comparing to $\tilde{E}$. In particular, we will show that for all $k$,
\begin{align}
		\label{eq:scalar_comparison}
\mathbb{E}\left[\|\tilde{\bv{e}}\|_2^{2k}\right] \leq \mathbb{E}\left[\tilde{E}^{2k}\right]. 
\end{align}
To do so, we compare the expansions:
\begin{align*}
	\mathbb{E}\left[\|\tilde{\bv{e}}\|_2^{2k}\right] &= \mathbb{E}\left[\left(\|\bv{w}_1\|_2^2 + \ldots + \|\bv{w}_m\|_2^2 + 2\bv{w}_1^T\bv{w}_2 + \ldots + 2\bv{w}_{m-1}^T\bv{w}_m\right)^k\right]\\
	\mathbb{E}\left[\tilde{E}^{2k}\right] &= \mathbb{E}\left[\left(W_1^2 + \ldots + W_m^2 + 2W_1W_2 + \ldots + 2W_{m-1}W_m\right)^k\right]
\end{align*}
Consider each term obtained when expanding out the $k^\text{th}$ powers above and apply linearity of expectation. Because each $\bv{w}_z$ is a symmetric random variable -- i.e. $\Pr(\bv{w}_z = X) = \Pr(\bv{w}_z = -X)$ -- we can verify that the expectation of any term where $\bv{w}_z^T\bv{w}_j$ appears an odd number of times (for \emph{any} fixed $j$) is equal to zero. Similarly, the corresponding term in the second sum has expectation zero because some $W_z$ must appear an odd number of times. For all other terms, we can use  that $\bv{w}_z^T\bv{w}_j \leq \|\bv{w}_z\|_2 \|\bv{w}_j\|_2$ (Cauchy–Schwarz) and that $\mathbb{E}[\|\bv{w}_z\|_2^{2k}] = \mathbb{E}[W_z^{2k}]$ to see that each term in the bottom expansion  upper bounds the corresponding term in the top. We conclude \cref{eq:scalar_comparison}. 

A Taylor expansion of $e^x$ combined with \cref{eq:moment_bound_1} and \cref{eq:scalar_comparison} implies a bound on the moment generating function (MGF) of $\|\bv{e}\|_2^2$, which we will use to obtain a final concentration result. Specifically, we have that for any $\lambda \geq 0$:
\begin{align}
	\label{eq:exp_moment_chain}
	\mathbb{E}\left[e^{\lambda\|\bv{e}\|_2^2} \right]\leq \mathbb{E}\left[e^{\lambda\|\tilde{\bv{e}}\|_2^2} \right]\leq \mathbb{E}\left[e^{\lambda\tilde{E}^2} \right].
\end{align}

%

\subsection{Moment Bound}
With \cref{eq:exp_moment_chain} in place, we prove our main result by bounding the exponential $\mathbb{E}\left[e^{\lambda\tilde{E}^2} \right]$ for our scalar random variable $\tilde{E}$. Specifically, we will show that $\tilde{E}$ is a sub-exponential random variable (\cref{def:subexp}), and thus $\|\bv{e}\|_2^2$ is as well by \cref{eq:exp_moment_chain}. We can then apply a standard tail bound for sub-exponential random variables.
\begin{proof}[Proof of \cref{thm:main}]
Recall that $\tilde{E} = \sum_{z=1}^m W_z$ is the sum of i.i.d. random variables. Recall that $W_z = r_z \cdot \|\bv{w}_z\|_2$, where $r_z$ is a random $\pm 1$ and $\bv{w}_z = \bv{e}_{z,1} - \bv{e}_{z,2}$. We have that  $\|\bv{w}_z\|_2^2 \leq  2\|\bv{e}_{z,1}\|_2^2 + 2\|\bv{e}_{z,2}\|_2^2$ and since $\bv{e}_{z,1}$ and $\bv{e}_{z,2}$ are just i.i.d. copies of $\bv{e}_z$, we have that for all $\lambda \geq 0$:
\begin{align}
	\label{eq:symmetrization}
	\mathbb{E}\left[e^{\lambda \|\bv{w}_z\|_2^2} \right] \leq 	\mathbb{E}\left[e^{4 \lambda \|\bv{e}_z\|_2^2} \right]. 
\end{align}
As discussed before, $\|\bv{e}_z\|_2^2$ is exactly equal to Hutchinson's estimator applied to the matrix $\bv{B} = \bar{\bv{A}}^T\bar{\bv{A}}$. Under the notation of \cref{lem:hutch_moment_bound}, $\|\bv{e}_z\|_2^2 = T(\bv{B}) = Z(\bv{B}) + \tr(\bv{B})$. We can thus apply \cref{lem:hutch_moment_bound} to obtain that for $0 \leq 4\lambda \leq {c}/{\|\bv{B}\|_2}$,
\begin{align*}
	\mathbb{E}\left[e^{4 \lambda \|\bv{e}_z\|_2^2} \right] =  e^{4\lambda \tr(\bv{B})}\mathbb{E}\left[e^{4 \lambda Z(\bv{B})} \right] \leq  e^{4\lambda \tr(\bv{B})}e^{16C \lambda^2\|\bv{B}\|_F^2}.
\end{align*}
Since $B = \bar{\bv{A}}^T\bar{\bv{A}}$ is a positive semidefinite matrix, $\|\bv{B}\|_F^2/\|\bv{B}\|_2 \leq \tr(\bv{B})$ and thus $4\lambda\|\bv{B}\|_F^2 \leq c\tr(\bv{B})$. Continuing we have:
\begin{align}
	\label{eq:subexp}
	\mathbb{E}\left[e^{4 \lambda \|\bv{e}_z\|_2^2} \right] \leq e^{4\lambda \tr(\bv{B}) + 4Cc\lambda\tr(\bv{B})}
\end{align}
We conclude from \cref{eq:symmetrization} and \cref{eq:subexp} that for all $\lambda \leq \frac{c}{4\|\bv{B}\|_2} \leq \frac{c}{4\tr(\bv{B})}$, 
	\begin{align}
		\label{eq:final_moment}
	\mathbb{E}\left[e^{\lambda \|\bv{w}_z\|_2^2} \right] \leq e^{c'\lambda \tr(\bv{B})},
	\end{align}
where $c' = 4 + 4Cc$ is a constant. 
\subsection{Completing the Proof }
Applying \cref{def:subexp}, we can check that \Cref{eq:final_moment} implies that $\|\bv{w}_z\|_2^2$ is a sub-exponential random variable with parameter $C'\cdot \tr(\bv{B})$, where $C' = \max(2c', 4/c)$. Equivalently, $\|\bv{w}_z\|_2$ is sub-Gaussian with parameter $K = \sqrt{C'}\sqrt{\tr(\bv{B})}$. 
Since it has the same moments as $\|\bv{w}_z\|_2$, $W_z = r_z \cdot \|\bv{w}_z\|_2$ is also sub-Gaussian with the same parameter. 

Proposition 2.6.1 from \cite{Vershynin:2018} states that the sum of $m$ mean $0$, independent, sub-Gaussian random variables, each with parameter $K$, is itself sub-Gaussian with parameter $C\cdot \sqrt{m} K$ for a fixed constant $C$.
We conclude that $\tilde{E} = \sum_{i=1}^m W_z$ is sub-Gaussian with parameter $C \sqrt{C'}\cdot \sqrt{m}\sqrt{\tr(\bv{B})}$. Finally, it follows that $\tilde{E}^2$ is sub-exponential with parameter $c''\cdot m\tr(\bv{B})$, where $c'' = C^2 C'$. From \cref{eq:exp_moment_chain}, we know that $\|\bv{e}\|_2^2$ is sub-exponential with the same parameter.
Finally, from Proposition 2.7.1 in \cite{Vershynin:2018} we have that $\Pr\left[\|\bv{e}\|_2^2 \geq t \right] \leq 2 e^{-\frac{t}{c''\cdot m \tr(\bv{B})}}$ and thus:
\begin{align*}
	\Pr\left[\frac{1}{m^2}\|\bv{e}\|_2^2 \geq \tr(\bv{B})\cdot \frac{c''\log(2/\delta)}{m} \right] \leq \delta.
\end{align*}
Recalling that $[\|\bv{r}^m(\bv{A}) - \diag(\bv{A})\|_2^2 = \frac{1}{m^2}\|\bv{e}\|_2^2$ and  $\tr(\bv{B}) = \|\bar{\bv{A}}\|_F^2$,  \cref{thm:main} follows.
\end{proof}

\section{General Stochastic Diagonal Estimators}
\label{sec:general_estimators}
In addition to the standard Hutchinson's estimator, prior work on stochastic diagonal and trace estimation also considers estimators involving Gaussian random vectors, or more generally, vectors filled with arbitrary mean $0$, variance $1$ random variables  \cite{Girard:1987,BastonNakatsukasa:2022}. 

\begin{definition}[Generalized Diagonal Estimator\footnote{We note that our generalized diagonal estimator differs slightly from some prior work, which considers the estimator $\left[\sum_{z=1}^m \bv{g}^{z}\odot \bv{A}\bv{g}^{z} \right] \oslash \left[  \sum_{z=1}^m \bv{g}^{z}\odot \bv{g}^{z} \right]$ \cite{BastonNakatsukasa:2022,BekasKokiopoulouSaad:2007}. The estimators only differ in the term used for scaling and for typical choices of random variables (e.g, Gaussains) should be very comparable, since every entry in $\sum_{z=1}^m \bv{g}^{z}\odot \bv{g}^{z}$ will concentrate closely around $m$. }] 
	\label{def:gen_diag_est}
	Let $\bv{g}^{1},\ldots, \bv{g}^{m} \in \R^n$ be independent random vectors, each with i.i.d. entries that have mean $0$ and variance $1$. The generalized stochastic diagonal estimator $\bv{d}^m(\bv{A})$ has the form:
	\begin{align*}
		\bv{d}^m(\bv{A}) = \frac{1}{m}\sum_{z=1}^m \bv{g}^{z}\odot \bv{A}\bv{g}^{z}.
	\end{align*}
\end{definition} 
When each $g^z_i$ has bounded 4th moment, we can prove a statement comparable to \cref{lem:expectation_lemma}.
\begin{lemma}[Expected Squared Error of Generalized Diagonal Estimator]
	\label{lem:general_expectation_lemma}
	Let $\bv{d}^m(\bv{A})$ be as in \cref{def:diag_est} and suppose each $g^z_i$ has 4th moment bounded by some constant $c_4$. I.e. $\E[(g^z_i)^4] \leq c_4$. Then we have:
	\begin{align*}
		\E\left[ \|\bv{d}^m(\bv{A}) - \diag(\bv{A}) \|_2^2\right] = \frac{1}{m}\left(\|\bar{\bv{A}}\|_F^2  + (1+ c_4 - 2)\sum_{i=1}^d A_{ii}^2 \right)
	\end{align*}
\end{lemma}
\begin{proof}
	As before, fix $z$ and  let $\bv{e}_z = \bv{g}^{z}\odot \bv{A}\bv{g}^{z} - \diag(\bv{A})$.
	Let $h_i = A_{ii} - A_{ii}(g_i^{z})^2$ and note that, since $\E(g_i^{z})^2 = 1$, we have that $\E h_i = 0$. More over, we have:
	\begin{align*}
		\E \left[h_i^2\right] = \E\left[ A_{ii}^2 + A_{ii}^2(g_i^{z})^4 - 2A_{ii}^2(g_i^{z})^2\right] =(c_4 + 1 - 2)A_{ii}^2.
	\end{align*}
	We then have that:
	\begin{align*}
		\|\bv{e}_z\|_2^2 = \sum_{i=1}^d \left(h_i + \sum_{j\neq i}A_{ij}g_i^{z}g_j^{z}\right)^2 &= \sum_{i=1}^d h_i^2 +  \sum_{i=1}^dh_i\sum_{j\neq i}A_{ij}g_i^{z}g_j^{z} +  \sum_{i=1}^d \left(\sum_{j\neq i}A_{ij}g_i^{z}g_j^{z}\right)^2 \\
		&= \sum_{i=1}^d h_i^2 +  \sum_{i=1}^dh_i\sum_{j\neq i}A_{ij}g_i^{z}g_j^{z} + \sum_{i=1}^d\sum_{j\neq i}\sum_{k \neq i} A_{ij}A_{ik} g_i^z g_i^zg_j^zg_k^z
	\end{align*}
	Considering each term separately, we can bound the expectation of $\|\bv{e}_z\|_2^2$. Noting that $\E [g_i^z g_i^zg_j^zg_k^z] = 0$ if $j\neq k$ and $1$ otherwise since $j\neq i$, we have:
	\begin{align*}
		\E \|\bv{e}_z\|_2^2  = \sum_{i=1}^d (1+ c_4  - 2) A_{ii}^2 + 0 +  \sum_{i=1}^d\sum_{j\neq i} A_{ij}^2 = (1+ c_4 - 2)  \|\diag(\bv{A})\|_2^2 + \|\bar{\bv{A}}\|_F^2.
	\end{align*}
\end{proof}
Combining \cref{lem:general_expectation_lemma} with Markov's inequality yields a simple dimension independent bound:
\begin{corollary}
	\label{cor:gen_simple}
	 Let $\bv{d}^m$ be implemented with any mean $0$ variance $1$ random variable with 4th moment upper bounded by $c_4$ and let $E = \sqrt{(1+ c_4 - 2)  \|\diag(\bv{A})\|_2^2 + \|\bar{\bv{A}}\|_F^2}$. For any $\delta \in (0,1)$ and $m\geq 1$, with probability $1-\delta$:
	\begin{align*}
		\|\bv{d}^m(\bv{A}) - \diag(\bv{A})\|_2 \leq \sqrt{\frac{1}{m\delta}}\cdot E.
	\end{align*}
\end{corollary}
\noindent When $\bv{d}^m$ is implemented with Gaussian random vectors, we have fourth moment $c_4 = 3$, so obtain the upper bound:
\begin{align*}
\|\bv{d}^m(\bv{A}) - \diag(\bv{A})\|_2 \leq \sqrt{\frac{2}{m\delta}}\cdot \|\bv{A}\|_F.
\end{align*} 

\subsection{High Probability Bounds}
\label{sec:general_high_prob}

To obtain an error bound of $\epsilon \cdot E$ with probability $1-\delta$, \cref{cor:gen_simple} requires $m = O\left({1}/{\epsilon^2\delta}\right)$ matrix-vector products with $\bv{A}$. The linear dependence on $1/\delta$ is worse than the logarithmic dependence in  \cref{thm:main}, which requires $m = O\left({\log(1/\delta)}/{\epsilon^2}\right)$ matrix-vector products for a comparable guarantee. It is possible to improve the dependence on $\delta$  using a high-dimensional analog of the standard ``median trick''. Specifically, we have:
\begin{corollary}
	\label{cor:median_trick}
	Consider the following estimation procedure that computes multiple independent generalized stochastic diagonal estimators (\cref{def:gen_diag_est}), all implemented with mean $0$ variance $1$ random variables with 4th moment $\leq c_4$. 
	\begin{itemize}
		\item Compute  $r = \lceil 10\log(1/\delta)\rceil$ independent generalized diagonal estimators $\bv{d}^{m}_1(\bv{A}), \ldots, \bv{d}^{m}_q(\bv{A})$.
		\item For all $i \in 1, \ldots, r$,  compute the distance $\|\bv{d}^{m}_i(\bv{A}) - \bv{d}^{m}_j(\bv{A})\|_2$ for all $j \neq i$. Let $B_i$ be the $\lfloor \frac{r}{2} \rfloor$ smallest distance.
		\item Return $\bv{d}^{m}_{i^*}(\bv{A})$, where $i^*=\argmin_{i\in 1,\ldots, r} B_i$.
	\end{itemize}
	There is an absolute constant $c$ so that, for any $\delta \in (0,1)$ and $m\geq 1$, with probability $1-\delta$:
	\begin{align*}
		\|\bv{d}^{m}_{i^*}(\bv{A}) - \diag(\bv{A})\|_2 \leq \sqrt{\frac{c}{m}}\cdot E,
	\end{align*}
where $E =\sqrt{(1+ c_4 - 2)  \|\diag(\bv{A})\|_2^2 + \|\bar{\bv{A}}\|_F^2}$, as before.
\end{corollary}
As desired, \cref{cor:median_trick} implies that $m = O\left({\log(1/\delta)}/{\epsilon^2}\right)$ matrix-vector multiplies are required to obtain error $\epsilon \cdot  E$ with probability $(1-\delta)$. 


\begin{proof}
	By \cref{cor:gen_simple}, for each $i \in 1, \ldots, r$ and a constant $c$, we have that $\|\bv{d}^{m}_i(\bv{A}) - \diag(\bv{A})\|_2 \leq c\sqrt{{1}/{m}}\cdot E$ with probability $19/20$. By a standard Chernoff bound argument, it follows that, with probability greater than $1 - e^{-r/10} = 1 - \delta$, $\|\bv{d}^{m}_i(\bv{A}) - \diag(\bv{A})\|_2 \leq c\sqrt{{1}/{m}}\cdot E$ for at least half of all values of $i$ (see e.g. Proposition 2.4 in \cite{AngluinValiant:1979}). Accordingly, by triangle inequality, we have that:
	\begin{align}
		\label{eq:median_1}
		\text{there is at least one } i \in 1, \ldots, r \text{ for which } B_i \leq 2c\sqrt{{1}/{m}}\cdot E.
	\end{align} 
	Also by pigeonhole principal, there must be at least one value of $j$ which is both one of the $\lfloor r/2\rfloor$ closest points to $\bv{d}^{m}_{i^*}$ and for which $\|\bv{d}^{m}_j(\bv{A}) - \diag(\bv{A})\|_2 \leq c\sqrt{{1}/{m}}\cdot E$.  I.e., there is some $j$ such that $\|\bv{d}^{m}_j(\bv{A}) - \diag(\bv{A})\|_2 \leq c\sqrt{{1}/{m}}\cdot E$ and $\|\bv{d}^{m}_j(\bv{A}) - \bv{d}^{m}_{i^*}(\bv{A}) \leq B_{i^*}$. By triangle inequality we thus have:
	\begin{align*}
		\|\bv{d}^{m}_{i^*}(\bv{A}) - \diag(\bv{A})\|_2 \leq c\sqrt{{1}/{m}}\cdot E + B_{i^*} \leq 3c\sqrt{{1}/{m}}\cdot  E.
	\end{align*}
The last inequality follows from \cref{eq:median_1} since $B_{i^*} \leq B_i$ for all $i \in 1, \ldots, r$. This completes the proof.
\end{proof}

If instead of just assuming that  $\bv{g}^1, \ldots, \bv{g}^m$ contain entries with bounded $4^\text{th}$ moment, if we make the stronger assumption that they contain i.i.d. sub-Gaussian entries, then we can obtain a bound for the generalized diagonal estimator that is more comparable to \cref{thm:main} and does not require the median trick to obtain a dependence the ideal dependence on $\log(1/\delta)$. Specifically, in \cref{app:subgauss}, we prove the following result: 
\begin{theorem}
	\label{thm:subgauss}
	Let $\bv{d}^m$ be a generalized stochastic diagonal estimator (\cref{def:gen_diag_est}) for $\bv{A}\in \R^{n\times n}$ implemented with any {symmetric}, mean $0$, and variance $1$ random variable that is sub-Gaussian with parameter $K$. Then with probability $1-\delta$:
	\begin{align*}
		\|\bv{d}^m(\bv{A}) - \diag(\bv{A})\|_2 \leq cK^2 \cdot \sqrt{\frac{\log(2/\delta)}{m} + \frac{\log^4(2/\delta)}{m^2} }\|\bv{A}\|_F.
	\end{align*}
\end{theorem}
Typically $K^2$ is a small constant (e.g. for the previously studied case of Gaussian random variables \cite{BastonNakatsukasa:2022}), so \cref{thm:subgauss} nearly matches \cref{thm:main}, except in two ways. First, as in \cref{cor:gen_simple}, it has a dependence on $\|\bv{A}\|_F^2$ instead of  $\|\bar{\bv{A}}\|_F^2$. In general, $\|\bar{\bv{A}}\|_F^2$ is always smaller. This is inherent: as shown in \cref{lem:general_expectation_lemma}, the expected error of $\bv{d}^m(\bv{A})$ has a dependence on $\|\bv{A}\|_F^2$ unless the fourth moment equals $1$, but this is only the case for $\pm 1$ Rademacher random variables. All other random variables with variance $1$ have higher $4^\text{th}$ moment.  

Second,  \cref{thm:subgauss} has an extra dependence on $\frac{\log^4(2/\delta)}{m^2}$ that does not appear in \cref{thm:main}. While this is a lower order term for large $m$ -- specifically, the bound matches \cref{thm:main} when $m \geq \log^{1.5}(1/\delta)$ -- we believe it can likely be improved or removed entirely, possibly by following a different proof technique.
  
\appendix

\section*{Acknowledgements}
We would like to thank Yuji Nakatsukasa, Eric Hallman, and Cameron Musco for helpful discussions. Christopher Musco was supported by NSF CAREER award No. 2045590. Prathamesh Dharangutte was supported by  NSF award No. CCF-2118953.

\bibliographystyle{apalike}
\bibliography{diag_est}


\section{General Sub-Gaussian Analysis}
\label{app:subgauss}
In this section, we focus on proving \cref{thm:subgauss} for general sub-Gaussian stochastic diagonal estimators. The proof follows a different approach and is more involved than our proof for Hutchinson's estimator in \cref{sec:main_proof}, which strongly relies on the fact that the estimator uses Rademacher random variables. 

\subsection{Initial Symmetrization}
We first show how  \cref{thm:subgauss} can be reduced to an equivalent statement involving a symmetric random vector:
\begin{lemma}
	\label{lem:subgauss_symmetric}
	Let $\bv{d}^{m}_1, \bv{d}^{m}_2$ be independent generalized stochastic diagonal estimators (\cref{def:gen_diag_est}) for $\bv{A}$ implemented with any {symmetric}, mean $0$, and variance $1$ random variable that is sub-Gaussian with parameter $K$.  For any $\delta \in (0,1)$ and $m\geq 1$, let $m = O(\log(1/\delta)/\epsilon^2)$. Then with probability $1-\delta$:
	\begin{align*}
		\|\bv{d}^{m}_1(\bv{A}) -  \bv{d}^{m}_2(\bv{A})\|_2 \leq cK^2 \cdot \sqrt{\frac{\log(2/\delta)}{m} + \frac{\log^4(2/\delta)}{m^2} }\|\bv{A}\|_F.
	\end{align*}
\end{lemma}
Before proving \cref{lem:subgauss_symmetric}, we show how it can be used to prove \cref{thm:subgauss}.
\begin{proof}[Proof of \cref{thm:subgauss}]
	Let $\ell = \log_2(1/\delta')$ for some $\delta'< \delta$ to be chosen later and consider independent diagonal estimators $\bv{d}^m_2, \ldots, \bv{d}^m_{\ell+1}$. We will not actually compute these estimators -- they are hypothetical and introduced for the purpose of analysis. Any random variable with sub-Gaussian parameter $K$ has fourth moment bounded by $O(K^4)$ (see \cite{Vershynin:2018}, Proposition 2.5.2). Accordingly, by \cref{cor:gen_simple}, for some constant $c$, we have that, with probability $1/2$, $\|\bv{d}^m_i(A) - \diag(\bv{A})\|_2 \leq c K^2 \|\bv{A}\|_F$ for each $i\in 2,\ldots, \ell+1$. 
	It follows that, with probability $\delta'$, $\|\bv{d}^m_j(\bv{A}) - \diag(\bv{A})\|_2 \leq  \frac{c K^2}{\sqrt{m}}\|\bv{A}\|_F$ for at least one value of $j$. At the same time, combining \cref{lem:subgauss_symmetric} with a union bound, we know that \emph{for all} $i$ simultaneously, with probability $1-\delta'\log_2(1/\delta')$, $\|\bv{d}^{m}_1(\bv{A}) -  \bv{d}^{m}_i(\bv{A})\|_2 \leq cK^2 \cdot \sqrt{\frac{\log(2/\delta')}{m} + \frac{\log^4(2/\delta')}{m^2} }\|\bv{A}\|_F.$
It follows by triangle inequality and another union bound that with probability
 $1-\delta'\log_2(1/\delta') - \delta '$, 
 \begin{align*}
 	\|\bv{d}^{m}_1(\bv{A}) -  \diag(\bv{A})\|_2 \leq \|\bv{d}^{m}_1(\bv{A}) -  \bv{d}^{m}_i(\bv{A})\|_2  + \|\bv{d}^{m}_i(\bv{A}) -  \diag(\bv{A})\|_2 \leq  2cK^2 \cdot \sqrt{\frac{\log(2/\delta)}{m} + \frac{\log^4(2/\delta)}{m^2} }\|\bv{A}\|_F.
 \end{align*} 
Setting $\delta' = c'\delta^2$ for sufficiently small constant $c'$ yields \cref{thm:subgauss}.
\end{proof}

\subsection{Single Sample Norm Bound}
In order to prove \cref{lem:subgauss_symmetric}, we first prove a tail bound on the norm of a single-sample sub-Gaussian stochastic diagonal estimator. This intermediate result is the crux of our analysis, and from it  \cref{lem:subgauss_symmetric} follows relatively directly.
\begin{lemma}
	\label{lem:gen_norm_bound} Let $\bv{e} = \bv{g}\odot \bv{A}\bv{g} - \diag(\bv{A})$, where $\bv{g} \in \R^n$ contains i.i.d. symmetric, mean $0$, and variance $1$ sub-Gaussian random variables with parameter $K$. For any $\gamma \geq 0$ and a fixed constant $c$ we have that
	\begin{align*}
	\Pr\left[\|\bv{e}\|_2^2 \geq \gamma K^4 \|\bv{A}\|_F^2 \right] \leq 2e^{-c\gamma^{1/3}}.
	\end{align*}
\end{lemma}
\begin{proof}
	In what follows, we will assume that  $\gamma \geq C$ for some sufficiently large constant $C$. If we can prove the result with some constant $c'$ in the exponent under this assumption, than we immediately have that $\Pr\left[\|\bv{e}\|_2^2 \geq \gamma K^4 \|\bv{A}\|_F \right] \leq 2e^{-c\gamma^{1/3}}$ for \emph{all} $\gamma \geq 0$, where $c=\min(c',1/2C^{1/3})$. This follows because $2e^{-\min(c',1/2C^{1/3})\gamma^{1/3}} > 1$ for any $\gamma \leq C$, so the bound is vacuously true for small values of $\gamma$.
	
	We start by applying triangle inequality and AM-GM inequality to give:
	\begin{align}
		\label{eq:initial_split}
		\|\bv{e}\|_2^2 \leq  2\|\bv{g}\odot \bv{A}\bv{g}\|_2^2 + 2\|\diag(\bv{A})\|_2^2 \leq 2\|\bv{g}\odot \bv{A}\bv{g}\|_2^2 + 2\|\bv{A}\|_F^2,
	\end{align}
	so we focus on bounding $\|\bv{g}\odot \bv{A}\bv{g}\|_2^2$. Following the proof of \cref{lem:general_expectation_lemma}, we have that:
	\begin{align*}
		\|\bv{g}\odot \bv{A}\bv{g}\|_2^2 = \sum_{i=1}^d \left(\sum_{j=1}^dA_{ij}g_ig_j\right)^2 =  \sum_{i=1}^d \sum_{j=1}^dA_{ij}g_ig_j\sum_{k=1}^dA_{ik}g_ig_k = \sum_{i=1}^d\sum_{j=1}^d\sum_{k=1}^d A_{ij}A_{ik}g_ig_ig_jg_k
	\end{align*}
Let $\bv{G}$ be a diagonal matrix containing $\bv{g}$ on its diagonal and let $\tilde{\bv{A}} = \bv{G}{\bv{A}}\bv{G}$.  
The matrix $\tilde{\bv{A}}$ has entries equal to $\tilde{A}_{ij} = A_{ij} g_ig_j$. Morever, note that, since each $g_i$ is assumed to by symmetric, it is identically distributed to $g_ir_i$ where $r_1, \ldots r_n$ are independent Rademacher random variables. So we equivalently have that:
\begin{align}
	\label{eq:tr_of_rand_a}
	\|\bv{g}\odot \bv{A}\bv{g}\|_2^2 =  \sum_{i=1}^d\sum_{j=1}^d\sum_{k=1}^d A_{ij}A_{ik}g_ig_ig_jg_kr_jr_kr_i^2= \sum_{i=1}^d\sum_{j=1}^d\sum_{k=1}^d \tilde{A}_{ij}\tilde{A}_{ik}r_jr_k = \bv{r}^T\bv{\tilde{A}}^T\bv{\tilde{A}}\bv{r}. 
\end{align}
We conclude that the quantity $\|\bv{g}\odot \bv{A}\bv{g}\|_2^2$ is exactly equal to Hutchinson's estimator (implemented with Rademacher random variables) applied to the matrix $\bv{B} = \tilde{\bv{A}}^T\tilde{\bv{A}}$. As such, we expect that $\|\bv{g}\odot \bv{A}\bv{g}\|_2^2$ will tightly concentrate around $\tr(\bv{B}) = \|\tilde{\bv{A}}\|_F^2$. So the main challenge becomes to bound $\|\tilde{\bv{A}}\|_F^2$, which itself is a random variable.

Fortunately, we can bound this quantity by again making a connection to trace estimation. We have that:
\begin{align*}
	\|\bv{\tilde{A}}\|_F^2 = \sum_{i=1}^d\sum_{j=1}^d {\bv{A}}_{ij}^2 (g_i)^2(g_j)^2 = {\bv{g}^2}^T ({\bv{A}}\circ{\bv{A}}) \bv{g}^2, 
\end{align*}
where $\bv{g}^2$ denotes the vector obtained by squaring each entry of $\bv{g}$.  Then, let $\bar{\bv{g}} = \bv{g}^2 - \bv{1}$ where $\bv{1}$ is an all ones vector. Note that $\E[\bar{\bv{g}}] = \bv{0}$ and since $\bv{g}$ is sub-Gaussian, $\bv{g}^2$ is a sub-exponential random variable with parameter $K^2$, and thus $\bar{\bv{g}}$ is sub-exponential with parameter $c'K^2$ for fixed constant $c'$ (see \cite{Vershynin:2018}, Exercise 2.7.10). We have that:
\begin{align}
	\label{eq:2_things_to_bound}
	{\bv{g}^2}^T ({\bv{A}}\circ{\bv{A}}) \bv{g}^2 = (\bar{\bv{g}}+\bv{1})^T({\bv{A}}\circ{\bv{A}})(\bar{\bv{g}}+\bv{1}) &= \bar{\bv{g}}^T({\bv{A}}\circ{\bv{A}})\bar{\bv{g}} + 2\bar{\bv{g}}^T({\bv{A}}\circ{\bv{A}})\bv{1} + \bv{1}^T({\bv{A}}\circ{\bv{A}})\bv{1} \nonumber \\
&= 	\bar{\bv{g}}^T({\bv{A}}\circ{\bv{A}})\bar{\bv{g}} + 2\bar{\bv{g}}^T({\bv{A}}\circ{\bv{A}})\bv{1} +\|\bv{A}\|_F^2.
\end{align}
We bound $\bar{\bv{g}}^T({\bv{A}}\circ{\bv{A}})\bar{\bv{g}}$ and $2\bar{\bv{g}}^T({\bv{A}}\circ{\bv{A}})\bv{1}$ separately,  starting with the second. Let $\bv{a}_i$ denote the $i^\text{th}$ row of $\bv{A}$ and note that $({\bv{A}}\circ{\bv{A}})\bv{1}$ has $i^\text{th}$ entry equal to $\|\bv{a}_i\|_2^2$. Since $\bar{\bv{g}}$ is mean $0$, we can apply a Bernstein inequality for sub-exponential random variables (\cite{Vershynin:2018}, Theorem 2.8.1) to the sum $\bar{\bv{g}}^T({\bv{A}}\circ{\bv{A}})\bv{1} = \sum_{i=1}^n\bar{g}_i\|\bv{a}_i\|_2^2$. We have that: 
\begin{align*}
\Pr\left[|\bar{\bv{g}}^T({\bv{A}}\circ{\bv{A}})\bv{1}| \geq tK^2\right] \leq 2\exp\left(-c''\min\left(\frac{t^2}{\sum_{i=1}^n \|\bv{a}_i\|_2^4}, \frac{t}{\max_{i} \|\bv{a}_i\|_2^2}\right)\right),
\end{align*}
where $c''$ is a fixed constant. Plugging in $t = \gamma\|\bv{A}\|_F^2$ and using that $\|\bv{A}\|_F^4 = \left(\sum_{i=1}^n \|\bv{a}_i\|_2^2\right)^2 \geq \sum_{i=1}^n \|\bv{a}_i\|_2^4$ and $\|\bv{A}\|_F^2 \geq \max_{i} \|\bv{a}_i\|_2^2$, we obtain the following bound fo any $\gamma \geq C$ for fixed constant $C$:
\begin{align}
	\label{eq:sub_exp_bern2}
	\Pr\left[|\bar{\bv{g}}^T({\bv{A}}\circ{\bv{A}})\bv{1}| \geq \gamma K^2 \|\bv{A}\|_F^2\right] \leq 2e^{-c\gamma}
\end{align}

Next we bound the $\bar{\bv{g}}^T({\bv{A}}\circ{\bv{A}})\bar{\bv{g}}$ term from \cref{eq:2_things_to_bound} using a Hanson-Wright type inequality for sub-exponential random variables due to \cite{GotzeSambaleSinulis:2021}.\footnote{The bound in \cite{GotzeSambaleSinulis:2021} is stated for \emph{symmetric} matrices, but it holds for all matrices without modification. In particular, for any $\bv{M}$,  $\bv{x}^T\bv{A}\bv{x} = \bv{x}^T\left(\frac{\bv{M} + \bv{M}^T}{2}\right)\bv{x}$, and by triangle inequality the symmetric matrix $\frac{\bv{M} + \bv{M}^T}{2}$ has Frobenius and spectral norm upper bounded by those of $\bv{A}$.} A similar bound is proven in \cite{Sambale:2020}.
\begin{lemma}
	\label{lem:general_hanson_wright}
	(Proposition 1.1 from \cite{GotzeSambaleSinulis:2021}) Let $\bv{x}$ be a random vector with i.i.d. mean $0$, variance $\sigma^2$ random entries that are sub-exponential with parameter $E$ and let $\bv{M}$ be any $n \times n$ matrix. For any $t > 0$ we have, 
	\begin{align*}
		\mathbb{P}\left(\left| \bv{x}^T\bv{M}\bv{x} - \sum_{i=1}^n \sigma^2 M_{ii} \right | \geq tE^2 \right) \leq 2 \exp \left( -c'' \min{\left( \frac{t^2}{\|\bv{M}\|_F^2}, \left(\frac{t}{\|\bv{M}\|_2}\right)^{1/2} \right)} \right).
	\end{align*}
\end{lemma}
To apply \Cref{lem:general_hanson_wright} to $\bar{\bv{g}}^T({\bv{A}}\circ{\bv{A}})\bar{\bv{g}}$, first note that $\bv{\bar{g}}$'s entries have variance $\sigma^2 \leq C K^4$ for some fixed constant $C$ because they are sub-exponential with parameter $c'K^2$.  So we have that $\sum_{i=1}^n \sigma^2 M_{ii} \leq CK^4 \sum_{i=1}^n A_{ii}^2 \leq CK^4 \|\bv{A}\|_F^2$.  Then plugging in $t = \frac{1}{c'^2}\gamma \|\bv{A}\|_F^2$ and using that $\|\bv{A}\|_F^4 = (\sum_{i,j}A_{ij}^2)^2 \geq \sum_{i,j}A_{ij}^4 = \|\bv{A}\circ \bv{A}\|_F^2 \geq \|\bv{A}\circ \bv{A}\|_2^2$,  we have that:
\begin{align*}
	\mathbb{P}\left(\left| \bv{\bar{g}}^T({\bv{A}}\circ{\bv{A}})\bv{\bar{g}} \right | \geq (\gamma K^4  + CK^4) \|\bv{A}\|_F^2 \right) \leq 2 e^{-c\gamma^{1/2}}, 
\end{align*}
for some fixed constant $c$ and any $\gamma \geq 0$. Under our assumption that $\gamma$ is larger than a fixed constant, we have that $CK^4 = O(\gamma K^4)$, so we can adjust the constant $c$ to simplify the expression to
\begin{align}
	\label{lem:sub_exp_hw2}
	\mathbb{P}\left(\left| \bv{\bar{g}}^T({\bv{A}}\circ{\bv{A}})\bv{\bar{g}} \right | \geq \gamma K^4 \|\bv{A}\|_F^2 \right) \leq 2 e^{-c\gamma^{1/2}}.
\end{align}

Plugging in  \cref{eq:sub_exp_bern2} and \cref{lem:sub_exp_hw2} to \cref{eq:2_things_to_bound} an applying a union bound, we conclude that:
\begin{align*}
\Pr[|{\bv{g}^2}^T ({\bv{A}}\circ{\bv{A}}) \bv{g}^2 | \geq  \gamma K^4 \|\bv{A}\|_F^2 +  2\gamma K^2 \|\bv{A}\|_F^2 + \|\bv{A}\|_F^2] \leq 4 e^{-c\gamma^{1/2}}.
\end{align*}
Since each entry of $\bv{g}$ has variance $1$, $K$ is greater than a fixed constant, so $\gamma K^4 \|\bv{A}\|_F^2 +  2\gamma K^2 \|\bv{A}\|_F^2 + \|\bv{A}\|_F^2 = O(\gamma K^4 \|\bv{A}\|_F^2)$. So again adjusting constants, we can simplify the above expression to claim that for any $\gamma \geq 0$,
\begin{align}
	\label{eq:final_tilde_A_bound}
	\Pr[\|\bv{\tilde{A}}\|_F^2  \geq  \gamma K^4 \|\bv{A}\|_F^2] \leq 2 e^{-c\gamma^{1/2}},
\end{align}
where we recall that $\|\bv{\tilde{A}}\|_F^2 = {\bv{g}^2}^T ({\bv{A}}\circ{\bv{A}}) \bv{g}^2$. 

We are now close to proving \cref{lem:gen_norm_bound}. To do so, we need to bound $\|\bv{g}\odot \bv{A}\bv{g}\|_2^2$, which as discussed, is exactly equal to Hutchinson's estimator applied to the positive semi-definite matrix $\bv{\tilde{A}}^T\bv{\tilde{A}}$ -- i.e., $\|\bv{g}\odot \bv{A}\bv{g}\|_2^2 = \bv{r}^T\bv{\tilde{A}}^T\bv{\tilde{A}}\bv{r}$ where $\bv{r}$ is a Rademacher random vector. Let $\bv{B}$ denote $\bv{B} = \bv{\tilde{A}}^T\bv{\tilde{A}}$. It follows from the Hanson-Wright inequality (see \cite{Vershynin:2018}, Theorem 6.2.1) that:
\begin{align*}
	\Pr\left[\left|\bv{r}^T\bv{B}\bv{r}-\tr(\bv{B})\right| \geq \gamma \|\bv{{B}}\|_F\right] \leq 2e^{-c\gamma}.
\end{align*}
Since $\bv{B}$ is PSD, we have that $ \|\bv{{B}}\|_F \leq \tr(\bv{B})$ and further we have that $\tr(\bv{B}) = \|\bv{\tilde{A}}\|_F^2$. So we can apply triangle inequality to conclude that $\Pr\left[\bv{r}^T\bv{B}\bv{r} \geq (\gamma + 1) \|\bv{\tilde{A}}\|_F^2\right] \leq 2e^{-c\gamma}$. Adjusting constants, it follows that for any $\gamma$,
\begin{align}
		\label{eq:almost_final_bound1}
	\Pr\left[\|\bv{g}\odot \bv{A}\bv{g}\|_2^2 \geq \gamma \|\bv{\tilde{A}}\|_F^2\right] \leq 2e^{-c\gamma}.
\end{align}
We combine this bound with \cref{eq:final_tilde_A_bound} to conclude that:
\begin{align}
	\label{eq:almost_final_bound}
	\Pr\left[\|\bv{g}\odot \bv{A}\bv{g}\|_2^2 \geq \gamma K^4\|\bv{A}\|_F^2 \right] \leq 2e^{-c\gamma^{1/3}}.
\end{align}
To obtain \cref{eq:almost_final_bound}, observe that to have $\|\bv{g}\odot \bv{A}\bv{g}\|_2^2 \geq  \gamma K^4\|\bv{A}\|_F^2$ it must be that \emph{either} $\|\bv{\tilde{A}}\|_F^2 \geq \gamma^{2/3} K^4 \|\bv{A}\|_F^2$ or that $\|\bv{g}\odot \bv{A}\bv{g}\|_2^2 \geq \gamma^{1/3} \|\bv{\tilde{A}}\|_F^2$. By \cref{eq:final_tilde_A_bound}, the first event only happens with probability $\leq 2 e^{-c\gamma^{1/3}}$ and by \cref{eq:almost_final_bound1} the second only happens with probability $\leq 2 e^{-c\gamma^{1/3}}$. Adjusting constants gives the equation. 

Finally, we return to equation \cref{eq:initial_split}, combining it with \cref{eq:almost_final_bound} to conclude that:
\begin{align*}
		\Pr\left[\|\bv{e}\|_2^2 \geq (2\gamma K^4 + 1)\|\bv{A}\|_F^2 \right] \leq 2e^{-c\gamma^{1/3}}.
\end{align*}
Again, since $K$ is greater than a fixed constant, we have that $\gamma K^4  = \Omega(1)$ and adjusting constants yields \cref{lem:gen_norm_bound}.
\end{proof}

\subsection{Completing the Proof}

We are finally ready to prove \cref{lem:subgauss_symmetric}, which we do by taking advantage of the symmetry of $\bv{d}^{m}_1(\bv{A}) -  \bv{d}^{m}_2$. Our proof uses a standard version of McDiamard's inequality (see e.g. \cite{Vershynin:2018}, Theorem 2.9.1), which we state below:
\begin{fact}[McDiarmid’s Inequality]
	\label{fact:mcdiamards}
	Let ${x}_1, \ldots, {x}_m \in \mathcal{X}_1 \times \ldots \times \mathcal{X}_m$ be independent random variables from domains $\mathcal{X}_1, \ldots, \mathcal{X}_m$.  Let $f:  \mathcal{X}_1 \times \ldots \times \mathcal{X}_m \rightarrow \R$ be any function such that for each coordinate $i$ and all realizations of ${x}_1, \ldots, {x}_m$, we have a difference bound of $\max_{\tilde{{x}}_i\in \mathcal{X}_i} \left | f({x}_1, \ldots, {x}_i, \ldots, {x}_m) - f({x}_1, \ldots, \tilde{{x}}_i, \ldots, {x}_m)\right | \leq c_i$.
	Then for any $t > 0$, 
	\begin{align*}
		\Pr\left[\left|f(\bv{x}_1, \ldots, \bv{x}_m) - \E f(\bv{x}_1, \ldots, \bv{x}_m)\right| \geq t\right] \leq 2e^{-\frac{2t^2}{\sum_{i=1}^m c_i^2}}.
	\end{align*}
\end{fact}


\begin{proof}[Proof of \cref{lem:subgauss_symmetric}]
	For $z\in 1,\ldots, n$ and $i\in 1,2$, let $\bv{e}_{z,i}$ be a random variable distributed as $\bv{g}\odot \bv{A}\bv{g} - \diag(\bv{A})$. Let $\bv{w}_z = \bv{e}_{z,1} - \bv{e}_{z,2}$ and let $r_1, \ldots, r_m$ be i.i.d Rademacher random variables. By the symmetry of each $\bv{w}_z$, we can write:
	\begin{align*}
		\bv{d}^m_1 -\bv{d}^m_2 = \frac{1}{m}\sum_{z=1}^m r_z \bv{w}_z.
	\end{align*}
We condition on the random choice of $\bv{w}_1, \ldots, \bv{w}_z$ and apply McDiamard's inequality. Specifically, by triangle inequality, the function $f(r_1, \ldots, r_m) =  \left\|\frac{1}{m}\sum_{z=1}^m r_z \bv{w}_z\right\|_2 = 	\|\bv{d}^{m}_1(\bv{A}) -  \bv{d}^{m}_2\|_2$ can change by at most $2\|\bv{w}_z\|_2/m$ if we change the input $r_z$. So by \cref{fact:mcdiamards}, we have that:
\begin{align*}
	\Pr\left[\left| \left\|\bv{d}^m_1 -\bv{d}^m_2\right\|_2 - \E\left[ \left\|\bv{d}^m_1 -\bv{d}^m_2\right\|_2 \right]\right|\geq t\right] \leq 2e^{-\frac{m^2t^2}{2\sum_{z=1}^m \|\bv{w}_z\|_2^2}}.
\end{align*}
By triangle inequality, we have that $\E\left[ \|\bv{d}^m_1 -\bv{d}^m_2\|_2\right] \leq \E\left[ \|\bv{d}^m_1\|_2\right] + \E\left[ \| \bv{d}^m_2\|_2\right]$. Moreover, by \cref{lem:general_expectation_lemma}, and the fact that $\E[X]^2 \leq \E[X^2]$ for any random variable $X$, we have that $\E\left[\|\bv{d}^m_1\|_2\right] \leq cK^2\|\bv{A}\|_F/\sqrt{m}$ for a fixed constant $c$. So plugging in $t  =  \frac{1}{m}\sqrt{2\gamma\sum_{z=1}^m \|\bv{w}_z\|_2^2}$, overall we conclude that for any $\gamma \geq 0$,
\begin{align}
		\label{eq:mcdiamards_conclusion}
	\Pr\left[\left\|\bv{d}^m_1 -\bv{d}^m_2\right\|_2  \geq \sqrt{\frac{1}{m}}\left(K^2\|\bv{A}\|_F +  \sqrt{2\gamma\sum_{z=1}^m \|\bv{w}_z\|_2^2/m} \right)\right] \leq 2e^{-\gamma}.
\end{align}

With \cref{eq:mcdiamards_conclusion} in place, we are left to bound $\sum_{z=1}^m \|\bv{w}_z\|_2^2$. By triangle inequality, this sum can be upper bounded $2\sum_{z=1}^m \|\bv{e}_{z,1}\|_2^2 + \|\bv{e}_{z,2}\|_2^2$. By  \cref{lem:gen_norm_bound}, each $\bv{e}_{z,i}$ satisfies $\Pr\left[\|\bv{e}_{z,i}\|_2^2 \geq \gamma K^4 \|\bv{A}\|_F^2 \right] \leq 2e^{-c\gamma^{1/3}}$. So, following the characterization of generalized subexponential random variables from \cite{Sambale:2020} (see Proposition 5.1 in that work), we conclude that for a constant $c$, $\|\bv{e}_{z,i}\|_2^2$ is an $\alpha$-subexponential random variable\footnote{Note that this is different from a subexponential random variable with parameter $\alpha$, as in \cref{def:subexp}. An $\alpha$-subexponential random variable as defined by \cite{GotzeSambaleSinulis:2021,Sambale:2020} has slower asymptotic tail decay than a standard subexponential random variable when $\alpha < 1$.} for $\alpha = 1/3$, with parameter $cK^4 \|\bv{A}\|_F^2$. Applying Lemma A.3 from \cite{GotzeSambaleSinulis:2021}, we have that $\|\bv{e}_{z,i}\|_2^2 - \E[\|\bv{e}_{z,i}\|_2^2]$ is also $\frac{1}{3}$ subexponential  with parameter $c'K^4 \|\bv{A}\|_F$. We can then apply Corollary 1.4 from \cite{GotzeSambaleSinulis:2021} to conclude that for all $\beta \geq 0$,
\begin{align*}
	\Pr\left(\left|\sum_{z=1}^m\sum_{i=1,2} \|\bv{e}_{z,i}\|_2^2 - \E\left[\|\bv{e}_{z,i}\|_2^2\right]\right|  \geq m\cdot \beta K^4 \|\bv{A}\|_F^2\right) \leq 2e^{-c\min\left(\beta m, \beta^{1/3}m^{1/3}\right)}.
\end{align*}
By \cref{lem:general_expectation_lemma} we have that $\E\left[\|\bv{e}_{z,i}\|_2^2\right] \leq \frac{C}{2}K^4 \|\bv{A}\|_F^2$ for all $i,z$ and a constant $C$. So, applying triangle inequality, adjusting constants, and recalling that $\sum_{z=1}^m \|\bv{w}_z\|_2^2 \leq 2\sum_{z=1}^m\sum_{i=1,2} \|\bv{e}_{z,i}\|_2^2$, we conclude that: 
\begin{align}
	\label{eq:final_sum_bound}
	\Pr\left(\sum_{z=1}^m \|\bv{w}_z\|_2^2 \geq m\cdot (1 + \beta) K^4 \|\bv{A}\|_F^2\right) \leq 2e^{-c \beta^{1/3}m^{1/3}}.
\end{align}
Combining \cref{eq:final_sum_bound} with \cref{eq:mcdiamards_conclusion} and again adjusting constants we have that for constants $C,c$,
\begin{align*}
	\label{eq:final_sum_bound}
		\Pr\left[\left\|\bv{d}^m_1 -\bv{d}^m_2\right\|_2  \geq		 \sqrt{\frac{1}{m}}\left(1+\sqrt{\gamma(1 + \beta )}\right)K^2\|\bv{A}\|_F 
		\right] \leq 2e^{-\gamma} + 2e^{-c\beta^{1/3} m^{1/3}}.
\end{align*}
The right hand side of the inequality is $\leq \delta$ as long as $\gamma \geq \log(4/\delta)$ and $\beta \geq \frac{\log^3(4/\delta)/c^3}{m}$. Plugging in and adjusting constants proves \cref{lem:subgauss_symmetric}. 


\end{proof}

\end{document}